\documentclass[11pt]{article}

\usepackage{geometry}
\usepackage{amsfonts}
\usepackage{graphicx}
\usepackage{epstopdf}
\usepackage{enumitem}
\usepackage[linesnumbered,ruled,vlined]{algorithm2e}
\usepackage{algorithmic}
\usepackage{amsmath,amsfonts,amssymb,mathrsfs}
\usepackage{amsthm}
\usepackage{subcaption}
\ifpdf
  \DeclareGraphicsExtensions{.eps,.pdf,.png,.jpg}
\else
  \DeclareGraphicsExtensions{.eps}
\fi

\usepackage{amsopn}
\usepackage{mathtools}

\theoremstyle{definition}
\newtheorem{definition}{Definition}[section]
\newtheorem{theorem}[definition]{Theorem}
\newtheorem{proposition}[definition]{Proposition}
\newtheorem{lemma}[definition]{Lemma}
\newtheorem{corollary}[definition]{Corollary}

\newtheorem*{claim}{Claim}

\DontPrintSemicolon
\DeclareMathOperator*{\argmin}{argmin}
\SetKwInput{KwInput}{Input} 
\SetKwInput{KwOutput}{Output}

\begin{document}

\title{An Algorithm for Linear Parametric Minimum Cycle Mean Problem}
\author{Yuki Nishida\thanks{Kyoto Prefectural University (E-mail: y-nishida@kpu.ac.jp)}    }

\date{}

\maketitle

\begin{abstract}
The minimum cycle mean problem (MCM) on weighted digraphs is the problem of finding the minimum value of the cycle mean, that is, the ratio of the cost to the length, over all cycles.
Despite its wide range of applications to discrete event systems, the parametric counterpart of the MCM has received relatively little attention in the literature, unlike other parametric problems in network optimization.
In this paper, we consider the linear parametric MCM, where all edges $e$ have cost $a(e)-b(e)t$ with parameter $t$.
We propose an algorithm to solve the linear parametric MCM in $O((m+n\log n)n^2W)$ time, where $n$ is the number of vertices, $m$ is the number of edges, and $W$ is the maximum absolute value of the coefficients $b(e) \in \mathbb{Z}$.
The central technique of the proposed method is the algorithm for the parametric shortest path problem.
The MCM is closely related to spectral theory in the tropical semiring, where the ``$\min$'' operation is regarded as addition and ``$+$'' as multiplication.
By exploiting the connection between them, we provide a method to compute the eigenvalues and eigenvectors of tropical parametric matrices.
\end{abstract}

\section{Introduction}

The minimum cycle mean problem (MCM) on a weighted digraph is the problem of finding the minimum value of the cycle mean, that is, the ratio of the cost to the length, over all cycles.
Karp~\cite {Karp1978} established a fundamental characterization of the minimum cycle mean. Since then, the MCM has been studied extensively because of its close connections to several areas of network optimization such as the negative cycle problem and the minimum cost flow problem.
The MCM also plays a crucial role in the analysis of discrete event systems, such as manufacturing~\cite{Cohen1985}, railway timetable~\cite{Heidergott2005}, and synchronous dataflow~\cite{Lee1987}, because the best performance of cyclic scheduling is given by the minimum cycle mean.

Linear parametric optimization is a class of optimization problems in which the objective function depends linearly on scalar parameters.
A central question in this setting is to understand the structure of the optimal value function and to determine how the optimal solution changes as the parameter varies.
Parametric optimization problems arise naturally in a wide range of applications, including sensitivity analysis~\cite{Gal1994}, process synthesis under uncertainty~\cite{Banerjee2003}, and model-predictive control~\cite{Pistikopoulos2012}. 
Parametric counterparts of fundamental network optimization problems, such as the shortest path problem, the assignment problem, and the minimum-cost flow problem, have also been studied~\cite{Nemesch2025}. 
Linear parametric optimization is also regarded as a multi-objective optimization because the parameter values serve as weights that determine the relative importance of the individual objective functions~\cite{Ehrgott2005}. 

Bringing these two lines of research together, we consider the parametric minimum cycle mean problem (parametric MCM). 
We deal with the case where there is only one linear parameter $t$, that is, the cost of each edge $e$ is of the form $a(e) - b(e)t$.  
To the best of our knowledge, the parametric MCM has received little attention even for this restricted setting.
Here, we define our problem more formally.
Let $\mathcal{G}=(V,E)$ be a digraph with $|V|=n$ and $|E|=m$.
Given two functions $a,b: E \to \mathbb{Z}$, define the parametric cost function $c(t): E \to \mathbb{R}$ by $c(t;e) = a(e)-b(e)t$ for all $e \in E$ and $t \in \mathbb{R}$.
By identifying a cycle $\mathcal{C}$ with the set of its edges, the cost of $\mathcal{C}$ is $c(t;\mathcal{C}):= \sum_{e\in \mathcal{C}} c(t;e)$ and the cycle mean of $\mathcal{C}$ is $\gamma(t;\mathcal{C}):= c(t;\mathcal{C})/|\mathcal{C}|$.
The problem is to find the breakpoints $t_0,t_1,\dots,t_N$, where $t_0=-\infty$ and $t_N = \infty$, and cycles $\mathcal{C}_1,\mathcal{C}_2,\dots,\mathcal{C}_N$ such that $\mathcal{C}_k$ minimizes $\gamma(t;\mathcal{C})$ over all cycles $\mathcal{C}$ if $t_{k-1} \leq t \leq t_k$ for $k = 1,2,\dots,N$.

\subsection{Our results}

Our main contribution in this paper is to devise an algorithm to solve the parametric MCM in $O((m+n\log n)n^2W)$ time, where $W:= \max_{e \in E} |b(e)|$.
We exploit the algorithm for the parametric shortest path problem by Karp and Orlin~\cite{Karp1981} and its improvement by Young et al.~\cite{Young1991} as a subroutine of the proposed algorithm.
Their algorithms update the shortest path tree incrementally as the parameter $t$ increases until a zero-cost cycle is detected.
Once a zero-cost cycle $\mathcal{C}$ is detected, the iteration terminates because a further increase of $t$ causes a negative-cost cycle, contradicting the existence of the shortest path tree.
Our idea for solving the parametric MCM is to continue the update of the shortest path tree by keeping the cost of $\mathcal{C}$ equal to zero.
To do that, we modify the parametric cost of all edges by subtracting $\gamma:= \alpha - \beta t$, which is the cycle mean of $\mathcal{C}$.
This modification prevents the cost of $\mathcal{C}$ from being negative, enabling us to search for a different cycle as $t$ increases.
By iteratively modifying the cost function each time a zero-cost cycle is found, the algorithm eventually obtains all breakpoints and corresponding minimum mean cycles.

A key observation for the complexity analysis is the number of breakpoints and the number of changes in the shortest path with respect to the modified cost function.
The former is $O(n^2W)$ because coefficients of $t$ in $\gamma(t;\mathcal{C})$ can take $O(n^2W)$ different values.
The latter is also $O(n^2W)$ for each pair of the starting and ending vertices of paths.
This is the most critical aspect of the analysis of the algorithm because the cost function used for the computation of shortest paths is dynamically updated during the algorithm.

With applications to scheduling of discrete event systems in mind, it is important to determine not only the minimum cycle mean but also the actual cyclic schedule of each event associated with the optimal cycle.
This has been well developed in the framework of the tropical semiring, which is also called the min-plus or max-plus algebra~\cite{Baccelli1992,Heidergott2005}.
The tropical semiring is a semiring of numbers $\mathbb{R} \cup \{\infty\}$ in which the ``$\min$'' operation is regarded as addition and ``$+$'' as multiplication.
In terms of the tropical semiring, the minimum cycle mean and the associated cyclic schedule are the minimum eigenvalue and its corresponding eigenvector of a square matrix, respectively.
It is well-known that the tropical eigenvector can be obtained by solving the shortest path problem~\cite{Green1979}.
Similarly to the tropical supereigenvector computation in~\cite{Nishida2026}, we utilize this result for the parametric case; the parametric eigenvectors are obtained successively as the shortest path tree is updated.

Notably, Karp's method~\cite{Karp1978} for the non-parametric MCM can be applied to the parametric MCM via symbolic computation over the tropical semiring using tropical polynomials in $t$.
The time complexity of this computation is $O(mn^2W)$.
This seems better than our proposed algorithm when $m$ is much smaller than $n \log n$.
For the non-parametric MCM, however, Young et al.~\cite{Young1991} remark that Karp's method always requires $O(mn)$ time of computation, whereas the time complexity $O(mn+n^2\log n)$ based on the parametric shortest path algorithm is only a worst-case bound, and suggest that the expected time complexity is $O(m+n\log n)$.
In addition, an experimental result in~\cite{Dasdan2004} shows that the parametric shortest path algorithm in~\cite{Young1991} is the fastest among six algorithms for solving the non-parametric MCM.
Hence, our proposed algorithm for the parametric MCM could be expected to offer strong theoretical guarantees and good practical performance.

\subsection{Related works}

Algorithms for parametric optimization on networks have been developed for several problems.
Karp and Orlin~\cite{Karp1981} presented an algorithm to solve the single-parametric shortest path problem in $O(mn \log n)$ time if the coefficients are $b(e) \in \{0,1\}$ for all $e \in E$.
Young et al.~\cite{Young1991} improved their algorithm to $O(mn + n^2\log n)$ time by using the Fibonacci heap data structure~\cite{Fredman1987}. They also state that the time complexity is $O((m + n\log n)nW)$ if $b(e) \in \{0,\pm 1, \dots, \pm W\}$.
Gassner and Klinz~\cite{Gassner2010} proposed an algorithm to solve the parametric assignment problem in $O((m + n\log n)nW)$ time by exploiting the parametric shortest path algorithm.
For the parametric minimum cost flow problem, Raith and Sede\~{n}o-Noda~\cite{Raith2017} presented an $O((m + n\log n)nB)$ time algorithm, where $B$ is the number of breakpoints.
Recently, the doubling constant parametrization for weighted network optimization problems has been investigated to evaluate the complexity of algorithms~\cite{Stoian2026}.

Compared with these problems, studies on the parametric MCM are quite limited.
In the context of the tropical eigenvalue problem, Pl\'{a}vka~\cite{Plavka2005} addressed a particularly special case of the parametric MCM in which the coefficients of parameters are $b(e) = 1$ for all edges $e$ leaving a fixed subset of vertices.
The mean-payoff game, which is a two-player game known to be in $\textrm{NP} \cap \textrm{co-NP}$, can be considered as a generalization of the MCM because it is equivalent to finding an eigenvalue of a min-max-linear map~\cite{Gaubert1998}.
Several studies have considered the parametric mean-payoff game as a tool for solving tropical optimization problems~\cite{Gaubert2012,Parsons2023}.
The parametric mean-payoff game has also been addressed as a reduction of the tropical generalized eigenvalue problem and solved using a naive approach that exhaustively solves the non-parametric problem over all candidate rational points~\cite{Gaubert2013}.
These studies consider parametric mean-payoff games only under specific settings corresponding to tropical optimization problem or generalized eigenvalue problem, and they do not completely contain the parametric MCM in this paper.
Parametric problems related to the MCM are also discussed in the literature on synchronous dataflow models~\cite{Bhattacharya2000}.
Ghamarian et al.~\cite{Ghamarian2008} considered the multi-parametric problem and applied a divide-and-conquer method to the enumeration of all convex regions on which the minimum cycle mean is linear.
Skelin et al.~\cite{Skelin2017} also considered the multi-parametric problem and analyzed the worst-case performance as the parameters vary over given ranges.
Similar approaches are found in the eigenvalue problem over tropical interval matrices~\cite{Cechlarova2005,Yin2022}.
In view of these circumstances, developing an efficient algorithm for the general single-parametric MCM appears to be a new challenge.

If we view the parametric MCM simply as a single-parametric optimization problem, general-purpose methods can be applied.
The Eisner-Severance method~\cite{Eisner1976} repeatedly solves the non-parametric problem to successively refine the piecewise-linear approximation of the parametric objective value.
When applying this method to the parametric MCM, the overall time complexity is $O(mn^3W)$ because the non-parametric MCM can be solved in $O(mn)$ time and the number of breakpoints is $O(n^2W)$, as noted previously.
Gusfield~\cite{Gusfield1983} proposed an algorithm based on Megiddo's method~\cite{Megiddo1979} for the single parametric search in combinatorial optimization.
Although this algorithm has a higher overall computational cost than the Eisner-Severance method, it has the advantage of being able to find the breakpoints one by one in polynomial time.

\subsection{Organization}

The rest of this paper is organized as follows.
In Section 2, we first introduce fundamental notations on graphs and then present the problem in consideration, the parametric MCM.
We also describe the algorithm for the parametric shortest path problem by Karp and Orlin~\cite{Karp1981} and Young et al.~ \ cite {Young1991}.
In Section 3, we present our main algorithm solving the parametric MCM and theoretical analysis of it.
In Section 4, we briefly introduce the tropical semiring and then present an application of our algorithm to computing the tropical eigenvectors of a parametric matrix.
In Section 5, we compare our proposed algorithm to the symbolic variant of Karp's method~\cite{Karp1978}.
Finally, we conclude the paper with some remarks in Section 6.

\section{Preliminaries}

In this paper, let $\mathbb{Z}_{> 0}$ and $\mathbb{Z}_{\geq 0}$ denote the sets of positive and nonnegative integers, respectively.
For $n \in \mathbb{Z}_{> 0}$, let $[n]$ denote the set $\{1,2,\dots,n\}$.
For $a, b \in \mathbb{R} \cup \{-\infty,\infty\}$ with $a \leq b$, the symbols $[a,b]$ and $[a,b)$ denote the closed and half-open intervals between $a$ and $b$, that is, $[a,b] = \{ t \in \mathbb{R} \mid a \leq t \leq b\}$ and $[a,b) = \{ t \in \mathbb{R} \mid a \leq t < b\}$.

Let $\mathcal{G} = (V, E)$ be a digraph with the vertex set $V$ and edge set $E$.
We allow self-loops on vertices and parallel edges.
For an edge $e \in E$, let $e^-$ and $e^+$ denote the tail and head of $e$, respectively.
For a vertex $v \in V$, let $\delta^+(v)$ and $\delta^-(v)$ denote the set of edges leaving $v$ and entering $v$, respectively, that is, $\delta^+(v) = \{e \in E \mid e^- = v\}$ and $\delta^-(v) = \{e \in E \mid e^+ = v\}$.
The outdegree and indegree of $v \in V$ are defined by $\deg^+(v) := |\delta^+(v)|$ and $\deg^-(v) := |\delta^-(v)|$, respectively.

A walk in a digraph $\mathcal{G} = (V,E)$ is an alternating sequence of vertices and edges $(v_0,e_1,v_1,e_2,\dots,e_\ell,v_\ell)$, where $\ell \in \mathbb{Z}_{\geq 0}$, $v_0,v_1,\dots,v_\ell \in V$, and $e_1,e_2,\dots,e_\ell \in E$, such that $v_{k-1} = e_k^-$ and $v_{k} = e_k^+$ for $k=1,2,\dots,\ell$.
The vertices $v_0$ and $v_\ell$ are called the starting and ending vertices of the walk, respectively.
A walk starting with $v_0$ and ending with $v_\ell$ is, for short, called a $v_0$-$v_\ell$ walk.
A digraph is said to be strongly connected if there exists a $u$-$v$ walk for any vertices $u$ and $v$.
A walk $(v_0,e_1,v_1,\dots,v_\ell)$ is called a path if $v_k \neq v_{k'}$ for $0 \leq k < k' \leq \ell$ and called a cycle if $v_0 = v_\ell$ and $v_k \neq v_{k'}$ for all pairs $(k,k') \neq (0,\ell)$ with $k < k'$.
We assume that a sequence $(v_0)$ is a $v_0$-$v_0$ path, while $(v_0,e_0,v_0)$ is a cycle defined by the self-loop $e_0$ on $v_0$.
We will identify a path or cycle $(v_0,e_1,v_1,\dots,v_\ell)$ with its edge set $\{e_1,e_2,\dots,e_\ell\}$.

Fix any vertex $r \in V$. An $r$-rooted tree is a digraph $\mathcal{T}$ such that $\deg^+(v) = 1$ for all $v \in V \setminus \{r\}$ and it has no cycles.
We identify a tree with its edge set.
In an $r$-rooted tree $\mathcal{T}$, there exists a unique $v$-$r$ path for each $v \in V$.
This path is denoted as $\mathcal{T}_v$.
For $u,v \in V$, the unique $u$-$v$ path in $\mathcal{T}$ is denoted as $\mathcal{T}_{[u,v]}$ if it exists.
For $v \in V$, the subgraph of $\mathcal{T}$ induced by the vertex set $\{ u \in V \mid \text{there exists $u$-$v$ path in $\mathcal{T}$}\}$ is denoted as $\mathcal{T}_{\to v}$.

\subsection{Parametric minimum cycle mean problems}

Let us consider a digraph $\mathcal{G} = (V, E)$ and a cost function $c: E \to \mathbb{R}$.
The cost of a path (or cycle) $\mathcal{P}$ is defined by $c(\mathcal{P}) := \sum_{e \in \mathcal{P}} c(e)$.
More generally, for any function $f: E \to \mathbb{R}$ and subset $E' \subset E$, we define $f(E') := \sum_{e \in E'} f(e)$.
The length of a path (or cycle) $\mathcal{P}$ is defined by $|\mathcal{P}|$, that is, the number of edges contained in the path (or cycle).
The cycle mean of a cycle $\mathcal{C}$ is then defined by $\gamma(\mathcal{C}) := c(\mathcal{C})/|\mathcal{C}|$.
The minimum cycle mean problem (MCM) we focus on in this paper is to find the minimum value of  $\gamma(\mathcal{C})$ over all cycles $\mathcal{C}$ in the given digraph $\mathcal{G} = (V, E)$, as well as the cycle attaining the minimum.
The following formula by Karp~\cite{Karp1978} computes the minimum cycle mean $\gamma$ in $O(mn)$ time:
\begin{align}
	\gamma = \min_{v \in V} \max_{\ell=0,1,\dots,n-1} \frac{F_v(n) - F_v(\ell)}{n-\ell},		\label{eq:mcmkarp}
\end{align}
where $F_v(\ell)$ is the minimum cost of all walks from $v$ to a fixed vertex $r \in V$ with length $\ell$.

Given two functions $a, b: E \to \mathbb{Z}$, we consider a linear parametric cost function $c = a-bt$ defined by $c(t;e):= a(e) - b(e)t$ for $e \in E$.
Here, $t$ is a parameter that can take any real value.
When specifying the value of $t$, we write the cost function as $c(t)$.
The parametric MCM is the problem of finding the breakpoints $t_1,t_2,\dots,t_{N-1}$ and cycles $\mathcal{C}_1,\mathcal{C}_2,\dots,\mathcal{C}_N$ such that 
\begin{align*}
	\gamma(t) := \min_{\mathcal{C}: \text{ cycle}} \frac{c(t;\mathcal{C})}{|\mathcal{C}|}
	=\frac{c(t;\mathcal{C}_k)}{|\mathcal{C}_k|}  \quad \text{if } t \in [t_{k-1},t_k],
\end{align*}
for $k \in [N]$, where $t_0 = -\infty$ and $t_N = \infty$.

\subsection{Parametric shortest path algorithms}

Let us consider a digraph $\mathcal{G} = (V, E)$ and a cost function $c: E \to \mathbb{R}$.
For $u,v \in V$, a $u$-$v$ path is called the shortest $u$-$v$ path if it has the minimum cost among all $u$-$v$ walks. Note that the minimum is always attained by a path if $\mathcal{G}$ does not have any negative-cost cycles; otherwise, the minimum does not exist.
An $r$-rooted tree consisting of the shortest $v$-$r$ path for all $v \in V$ is called the shortest path tree.
To simplify the discussion in Section 4, we consider the single-sink shortest path problem instead of the single-source one.

Now, we again introduce two functions $a, b: E \to \mathbb{Z}$ which define the parametric cost function $c= a -bt$.
The parametric (single-sink) shortest path problem is the problem to find breakpoints $t_0,t_1,\dots,t_{N}$ and $r$-rooted trees $\mathcal{T}_1,\mathcal{T}_2,\dots,\mathcal{T}_N$ with a fixed vertex $r \in V$ such that, for $k \in [N]$, the tree $\mathcal{T}_k$ is the shortest path tree with respect to the cost function $c(t)$ for all $t \in [t_{k-1},t_k]$.

Here, we describe the parametric shortest path algorithm by Karp and Orlin~\cite{Karp1981} and its improvement by Young et al.~\cite{Young1991}, on which our algorithm for the parametric MCM is based.
Suppose that the $r$-rooted shortest path tree $\mathcal{T}$ with respect to $c(t)$ is given for some $t \in \mathbb{R}$.
We investigate when the shortest path tree changes as $t$ increases.
For $v \in V$, consider any $v$-$r$ path $\mathcal{P}$ different from $\mathcal{T}_v$.
If $b(\mathcal{P}) \leq b(\mathcal{T}_v)$, then the cost $c(t;\mathcal{P})$ decreases no more than $c(t;\mathcal{T}_v)$ as $t$ increases.
Hence, $\mathcal{P}$ will never become the shortest $v$-$r$ path.
On the other hand, if $b(\mathcal{P}) > b(\mathcal{T}_v)$, then there exists $t' > t$ such that $c(t';\mathcal{P}) = c(t';\mathcal{T}_v)$.
When finding the smallest value of such $t'$ over all paths, we may only consider the paths of the form $\mathcal{P} = \{e\} \cup \mathcal{T}_{e^+}$ for $e \in \delta^+(v)$.
The value $t'$ such that $c(t;\{e\} \cup \mathcal{T}_{e^+}) = c(t;\mathcal{T}_v)$ is computed as
\begin{align*}
	t' = \frac{a(e) + a(\mathcal{T}_{e^+}) - a(\mathcal{T}_v)}{b(e) + b(\mathcal{T}_{e^+}) - b(\mathcal{T}_v)}.
\end{align*}
Combining the above two cases, the key $\kappa_{c,\mathcal{T}}(e)$ of $e \in E$ with respect to the parametric cost function $c = a-bt$ and tree $\mathcal{T}$ is defined by
\begin{align*}
	\kappa_{c,\mathcal{T}}(e) := 
	\begin{cases} \dfrac{a(e) + a(\mathcal{T}_{e^+}) - a(\mathcal{T}_{e^-})}{b(e) + b(\mathcal{T}_{e^+}) - b(\mathcal{T}_{e^-})}
	& \text{if } b(e) + b(\mathcal{T}_{e^+}) > b(\mathcal{T}_{e^-}), \\
	\infty & \text{otherwise}. \end{cases}
\end{align*}
Then, $\min_{e \in E} \kappa_{c,\mathcal{T}}(e)$ is the minimum value of $t$ at which the shortest path tree changes.

\begin{proposition}[\cite{Karp1981}]		\label{prop:psp}
For any fixed value $t^* \in \mathbb{R}$, let $\mathcal{T}$ be the $r$-rooted shortest path tree with respect to $c(t^*)$ and $\kappa = \min_{e \in E} \kappa_{c,\mathcal{T}}(e)$.
Then, $\mathcal{T}$ is also the $r$-rooted shortest path tree with respect to $c(t)$ for all $t \in [t^*, \kappa]$.
\end{proposition}

Let us consider the situation where $t$ is increased to $\kappa = \min_{e \in E} \kappa_{c,\mathcal{T}}(e)$ and $e \in E$ is the edge attaining the minimum key value.
Let $u = e^+$ and $v = e^-$. If $v$ is not on $\mathcal{T}_{u}$, then we replace the unique edge $e' \in \delta^+(v) \cap \mathcal{T}$ with $e$, obtaining the new shortest path tree $\mathcal{T}_{\mathrm{new}}$ with respect to $c(\kappa)$.
The key values are updated with respect to $\mathcal{T}_{\mathrm{new}}$.

On the other hand, if $v$ is on $\mathcal{T}_{u}$, then $\{e\} \cup \mathcal{T}_{[u,v]}$ becomes a zero-cost cycle with respect to $c(\kappa)$, which then has a negative cost with respect to $c(t)$ for $t > \kappa$.
Hence, there exists no shortest path tree with respect to $c(t)$ for $t > \kappa$, determining $t = \kappa$ as the final breakpoint.
The parametric shortest path algorithm by Karp and Orlin~\cite{Karp1981} is summarized as follows.
\begin{enumerate}
\item Take sufficiently small $t_0$. Find the shortest path tree $\mathcal{T}$ with respect to $c(t_0)$ and compute the corresponding keys $\kappa_{c,\mathcal{T}}(e)$ for all $e \in E$.
\item Find the minimum key value $\kappa$ as the next breakpoint. 
\item If a zero-cost cycle is detected, the algorithm terminates.
Otherwise, update the shortest path tree $\mathcal{T}$ and keys $\kappa_{c,\mathcal{T}}(e)$ for all $e \in E$ and repeat the procedure. 
\end{enumerate}

Young et al.~\cite{Young1991} improved the algorithm in~\cite{Karp1981} by introducing the vertex key
\begin{align*}
	\kappa_{c,\mathcal{T}}(v) := \min_{e \in \delta^+(v)} \kappa_{c,\mathcal{T}}(e), \quad v \in V,
\end{align*}
and using the Fibonacci heap data structure~\cite{Fredman1987} to maintain the key values. 
Let $n = |V|, m=|E|$, and $W = \max_{e \in E}|b(e)|$.
The path $\mathcal{T}_v$ changes at most $2(n-1)W$ times because the value of $b(\mathcal{T}_v)$ can take integer values between $-(n-1)W$ and $(n-1)W$ and must increase when $\mathcal{T}_v$ changes.
Using this fact, they reduced the time complexity to $O((m + n\log n)nW)$. 

We note that the parametric shortest path algorithm can be used to solve the non-parametric MCM.
To do that, we set $b(e) = 1$ for all $e \in E$. 
If the zero-cost cycle is detected at $t = \gamma$, then $\gamma$ is the minimum cycle mean with respect to the cost function $a$.
In this case, the time complexity of the algorithm by Young et al.~\cite{Young1991} is $O(mn+n^2\log n)$.
This is not better than Karp's method in~\eqref{eq:mcmkarp}.
However, it is suggested in~\cite{Young1991} that this worst-case bound is not achieved in practice and the expected time complexity is $O(m+n\log n)$, while Karp's method always requires $O(mn)$ computation time.

\section{Algorithms for parametric minimum cycle mean}

In this section, we present an algorithm to solve the parametric MCM based on the parametric shortest path algorithm.
Let $\mathcal{G}= (V, E)$ be a digraph and $c = a-bt$ be the parametric cost function determined by $a,b: E \to \mathbb{Z}$.
We may assume that $\mathcal{G}$ is strongly connected.
Indeed, every cycle in $\mathcal{G}$ is contained in some strongly connected component,
so we can solve the problem for each component separately.
For $\alpha, \beta \in \mathbb{R}$, we define functions $a_\alpha, b_\beta: E \to \mathbb{R}$ by
\begin{align*}
	a_{\alpha}(e) := a(e)-\alpha, \quad b_{\beta}(e) := b(e)-\beta, \quad e \in E.
\end{align*}
For $\gamma = \alpha - \beta t$, where $\alpha,\beta \in \mathbb{R}$, we define the modified cost function $c_{\gamma}: E \to \mathbb{R}$ by 
\begin{align*}
	c_{\gamma}(t;e) := a_{\alpha}(e) - b_{\beta}(e)t = (a(e)-b(e)t) - (\alpha - \beta t), \quad e \in E,\, t \in \mathbb{R}.
\end{align*}

The main part of the present algorithm, \textsc{ParametricMinimumCycleMean}, is shown in Algorithm~\ref{alg:main}.
We describe its outline.
First, we take a sufficiently small value $t_0$ and solve the non-parametric MCM with respect to the cost function $c(t_0) = a-bt_0$, finding the cycle $\mathcal{C}_1$ attaining the minimum.
We then begin the iteration by increasing the value of $t$.
The key idea is to keep the minimum cycle mean equal to zero throughout the algorithm.
To do that, we modify the cost function to 
\begin{align*}
	c_{\gamma_1} := a_{\alpha_1} - b_{\beta_1}t = c -(\alpha_1-\beta_1 t),\quad
	\text{where } \alpha_1 := \frac{a(\mathcal{C}_1)}{|\mathcal{C}_1|},\, \beta_1 :=  \frac{b(\mathcal{C}_1)}{|\mathcal{C}_1|},\,
	\gamma_1 := \alpha_1 - \beta_1 t.
\end{align*}
Then, we have $c_{\gamma_1}(t;\mathcal{C}_1) = 0$ for all $t \in \mathbb{R}$.
In addition, since $c_{\gamma_1}(t_0;e) = c(t_0;e) - \gamma_1(t_0)$ for all $e \in E$, the cycle mean of every cycle is reduced by the same amount $\gamma_1(t_0)$.
Hence, the minimum cycle mean with respect to $c_{\gamma_1}(t_0)$ is attained by $\mathcal{C}_1$.
For the modified cost function $c_{\gamma_1}$, the vertex keys are computed by using  $a_{\alpha_1}(\mathcal{T}_v)$ and $b_{\beta_1}(\mathcal{T}_v)$. These subroutines are shown in Algorithms~\ref{alg:tree} and~\ref{alg:updatekey}.
We find the breakpoint $t_1$ where another zero-cost cycle $\mathcal{C}_2$ with respect to $c_{\gamma_1}(t_1)$ appears.
This is done by solving the parametric shortest path problem for any fixed root $r \in V$.
Then, we again modify the cost function to 
\begin{align*}
	c_{\gamma_2} := a_{\alpha_2} - b_{\beta_2}t = c -(\alpha_2-\beta_2 t),\quad
	\text{where } \alpha_2 := \frac{a(\mathcal{C}_2)}{|\mathcal{C}_2|},\, \beta_2 :=  \frac{b(\mathcal{C}_2)}{|\mathcal{C}_2|},\,
	\gamma_2 := \alpha_2 - \beta_2 t,
\end{align*}
and solve the parametric shortest path problem to find the next breakpoint.
Since the cost functions $c_{\gamma_1}(t_1)$ and $c_{\gamma_2}(t_1)$ are identical, the shortest path tree $\mathcal{T}$ is unchanged.
The values $a_{\alpha_2}(\mathcal{T}_v)$ and $b_{\beta_2}(\mathcal{T}_v)$ must be computed at this point.
Repeating this process, we obtain the sequences of breakpoints $t_0,t_1,t_2,\dots,t_{N}$, minimum cycle means $\gamma_1,\gamma_2,\dots,\gamma_N$ in the corresponding intervals, and cycles $\mathcal{C}_1,\mathcal{C}_2,\dots,\mathcal{C}_N$ attaining the minima. 

\begin{algorithm}	\caption{\textsc{ParametricMinimumCycleMean}}	\label{alg:main}
    \KwIn{Digraph $\mathcal{G}=(V,E)$, cost functions $a,b: E \to \mathbb{Z}$}
    \KwOut{Breakpoints $t_0,t_1,\dots,t_{N} \in \mathbb{R} \cup \{-\infty,\infty\}$, linear forms $\gamma_1,\gamma_2,\dots,\gamma_N$ and cycles $\mathcal{C}_1,\mathcal{C}_2,\dots,\mathcal{C}_N$}
    $t_0 \gets -2n^2 \max_{e \in E} |a(e)| -1,\, k \gets 1$\;
    $\mathcal{C}_1 \gets \text{minimum mean cycle at $t=t_0$},\, \alpha_1 \gets \dfrac{a(\mathcal{C}_1)}{|\mathcal{C}_1|},\, \beta_1 \gets \dfrac{b(\mathcal{C}_1)}{|\mathcal{C}_1|},\, \gamma_1 \gets \alpha_1 - \beta_1 t$\;
    $\mathcal{T} \gets \text{$r$-rooted shortest path tree with respect to cost function $c_{\gamma_1}(t_0)$}$\;
    $(a_*,b_*) \gets \textsc{TreePath}(\mathcal{T},\alpha_1,\beta_1,a,b)$\;
    \For{$v \in V$}{
    	$(\kappa(v),\sigma(v)) \gets \textsc{UpdateKey}(\mathcal{G},v,a_*,b_*)$\;
    }
    \While{$\min_{v \in V} \kappa(v) < \infty$}{
    	$v^* \gets \argmin_{v \in V} \kappa(v),\, e^* \gets \sigma(v^*)$\;
    	\eIf{$v^*$ is on $\mathcal{T}_{(e^*)^+}$}{
    		$t_k \gets \kappa(v^*)$\;
    		$\mathcal{C}_{k+1} \gets \{e^*\} \cup \mathcal{T}_{[(e^*)^+,v^*]}$\;
    		$\alpha_{k+1} \gets \dfrac{a(\mathcal{C}_{k+1})}{|\mathcal{C}_{k+1}|},\, \beta_{k+1} \gets \dfrac{b(\mathcal{C}_{k+1})}{|\mathcal{C}_{k+1}|},\, \gamma_{k+1} \gets \alpha_{k+1} - \beta_{k+1} t$\;
    		$(a_*,b_*) \gets \textsc{TreePath}(\mathcal{T},\alpha_{k+1},\beta_{k+1},a,b)$\;
    		\For{$v \in V$}{
    			$(\kappa(v),\sigma(v)) \gets \textsc{UpdateKey}(\mathcal{G},v,a_*,b_*)$\;
    		}
    		$k \gets k+1$\;
    	}{
    		$e' \gets \text{unique edge in $\delta^+(v^*) \cap \mathcal{T}$}$\;
    		$\mathcal{T} \gets \mathcal{T} \cup e^* \setminus e'$\;
    		$\Delta a \gets a_*(e^*) + a_*((e^*)^+) - a_*(v^*),\, \Delta b \gets b_*(e^*) + b_*((e^*)^+) - b_*(v^*)$\;
    		\For{$v$ on $\mathcal{T}_{\to v^*}$}{
    			$a_*(v) \gets a_*(v) + \Delta a,\, b_*(v) \gets b_*(v) + \Delta b$\;
		}
		\For{$v$ on $\mathcal{T}_{\to v^*}$}{
    			$(\kappa(v),\sigma(v)) \gets \textsc{UpdateKey}(\mathcal{G},v,a_*,b_*)$\;
    			\For{$e \in \delta^-(v)$}{
    				\If{$b_*(e) + b_*(v) > b_*(e^-)$}{
    					$\kappa \gets  \dfrac{a_*(e) + a_*(v) - a_*(e^-)}{b_*(e) + b_*(v) - b^*(e^-)}$\;
    					\If{$\kappa < \kappa(e^-)$}{
    						$\kappa(e^-) \gets \kappa,\, \sigma(e^-) \gets e$\;
					}
   				}
    			}
    		}
    	}
    }
    $N \gets k,\, t_0 \gets -\infty,\, t_{N} \gets \infty$\;
    \KwRet $((t_0,t_1,\dots,t_N), (\gamma_1,\gamma_2,\dots,\gamma_N), (\mathcal{C}_1,\mathcal{C}_2,\dots,\mathcal{C}_N))$\;
\end{algorithm}

\begin{algorithm}	\caption{$\textsc{TreePath}(\mathcal{T},\alpha,\beta,a,b)$}		\label{alg:tree}
	\KwIn{$r$-rooted directed tree $\mathcal{T}$ on vertex set $V$, $\alpha, \beta \in \mathbb{R}$, functions $a,b: E \to \mathbb{Z}$}
	\KwOut{Functions $a_*,b_*: V \cup E \to \mathbb{R}$}
	\For{$e \in \mathcal{T}$}{
		$a_*(e) \gets a(e) - \alpha,\, b_*(e) \gets b(e) -\beta$\:
	}
	$S \gets \{r\},\, a_*(r) \gets 0,\, b_*(r) \gets 0$\;
	\While{$S \neq \emptyset$}{
		Take any $v \in S$\;
		\For{$e \in \delta^-(v) \cap \mathcal{T}$}{
			$a_*(e^-) \gets a_*(e) + a_*(v),\, b_*(e^-) \gets b_*(e) + b_*(v)$\;
			$S \gets S \cup \{e^-\}$\;
		}
		$S \gets S \setminus \{v\}$\;
	}
	\KwRet $(a_*,b_*)$\;
\end{algorithm}

\begin{algorithm}	\caption{$\textsc{UpdateKey}(\mathcal{G},v,a_*,b_*)$}	\label{alg:updatekey}
	\KwIn{Digraph $\mathcal{G}=(V,E)$, vertex $v \in V$, functions $a_*,b_*: V \cup E \to \mathbb{R}$}
	\KwOut{$\kappa \in \mathbb{R} \cup \{\infty\},\, \sigma \in E$}
	$\kappa \gets \infty$\;
	\For{$e \in \delta^+(v)$}{
		\If{$b_*(e) + b_*(e^+) > b_*(v)$}{
			$\kappa' \gets \dfrac{a_*(e) + a_*(e^+) - a_*(v)}{b_*(e) + b_*(e^+) - b_*(v)}$\;
			\If{$\kappa' < \kappa$}{
				$\kappa \gets \kappa',\, \sigma \gets e$\;
			}
		}
	}
	\KwRet $(\kappa,\sigma)$\;
\end{algorithm}

When we update the shortest path tree, we take the same approach as Young et al.~\cite{Young1991} to reduce the computational complexity.
We use the vertex keys for $v \in V$ that are maintained by the Fibonacci heap data structure.
If we replace an edge $e' \in \mathcal{T}$ with $e^*$, where $(e')^- = (e^*)^- = v^*$, the shortest paths $\mathcal{T}_v$ are changed for all $v$ on the subtree $\mathcal{T}_{\to v^*}$.
Hence, the keys of all such vertices $v$ need to be updated.
At the same time, the edge keys of $e$ such that $e^+ = v$ are decreased.
Any other keys remain unchanged.

\subsection{Theoretical analysis of Algorithm \ref{alg:main}}

Here, we prove the correctness of Algorithm~\ref{alg:main} and analyze its time complexity.
Let us define $n := |V|, m := |E|$, and $W := \max_{e \in E}|b(e)|$.
We may assume that $W \geq 1$; otherwise, the problem is just a non-parametric MCM.
The following lemmas for the subroutines are straightforward.

\begin{lemma}
Let $\mathcal{T}$ be an $r$-rooted tree on the vertex set $V$.
For $\alpha, \beta \in \mathbb{R}$ and functions $a,b: E \to \mathbb{Z}$, Algorithm~\ref{alg:tree} computes $a_\alpha(\mathcal{T}_v)$ and $b_\beta(\mathcal{T}_v)$ for all $v \in V$ in $O(n)$ time.
\end{lemma}

\proof
We have $a_*(v) =  a_{\alpha}(\mathcal{T}_v)$ and $b_*(v) = b_{\beta}(\mathcal{T}_v)$ for each $v \in V$ because the algorithm construct paths ending at $r$ by traversing $\mathcal{T}$ backward. As all vertices and edges in $\mathcal{T}$ are scanned exactly once, the time complexity is $O(n)$.
\endproof

\begin{lemma}		\label{lem:updatekey}
Let $\mathcal{T}$ be an $r$-rooted tree on the vertex set $V$.
For $\alpha, \beta \in \mathbb{R}$ and functions $a,b: E \to \mathbb{Z}$, suppose that $a_*(v) =  a_{\alpha}(\mathcal{T}_v)$ and $b_*(v) = b_{\beta}(\mathcal{T}_v)$ for all $v \in V$.
Let $\gamma = \alpha - \beta t$.
Then, for any $v \in V$, Algorithm~\ref{alg:updatekey} computes the vertex key $\kappa_{c_\gamma,\mathcal{T}}(v)$ and $e \in E$ attaining its minimum in $O(\deg^+(v))$ time.
\end{lemma}

\proof
By noting that $c_\gamma = a_\alpha - b_\beta t$, we have $\kappa' = \kappa_{{c_\gamma},\mathcal{T}}(e)$ if $b_\beta(e) + b_\beta(\mathcal{T}_{e^+}) > b_\beta(\mathcal{T}_v)$.
As all edges in $\delta^+(v)$ are scanned exactly once, the time complexity is $O(\deg^+(v))$.
\endproof

The initialization of Algorithm~\ref{alg:main} is justified by the following lemma.
\begin{lemma}		\label{lem:initial}
Let $t_0 =  -2n^2\max_{e\in E}|a(e)| -1$ and $\mathcal{C}_1$ be the minimum mean cycle with respect to the cost function $c(t_0)$.
Then, the minimum cycle mean with respect to $c(t)$ is attained by $\mathcal{C}_1$ for any $t \leq t_0$
\end{lemma}

\proof
On the contrary, suppose that the minimum cycle mean with respect to $c(t)$ is not attained by $\mathcal{C}_1$ for some $t < t_0$,
but by another cycle $\mathcal{C}$. 
Then, we have
\begin{align*}
	\frac{c(t_0;\mathcal{C}_1)}{|\mathcal{C}_1|} - \frac{c(t_0;\mathcal{C})}{|\mathcal{C}|} \leq 0 < \frac{c(t;\mathcal{C}_1)}{|\mathcal{C}_1|} - \frac{c(t;\mathcal{C})}{|\mathcal{C}|},
\end{align*}
yielding
\begin{align*}
	\left( \frac{b(\mathcal{C}_1)}{|\mathcal{C}_1|} - \frac{b(\mathcal{C})}{|\mathcal{C}|} \right) (t_0 - t) > 0.
\end{align*}
By noting that $t < t_0$, we have $\dfrac{b(\mathcal{C}_1)}{|\mathcal{C}_1|} - \dfrac{b(\mathcal{C})}{|\mathcal{C}|} > 0$.
Since both $\dfrac{b(\mathcal{C}_1)}{|\mathcal{C}_1|}$ and $\dfrac{b(\mathcal{C})}{|\mathcal{C}|}$ are rational numbers whose denominators are at most $n$, we have
\begin{align*}
	\dfrac{b(\mathcal{C}_1)}{|\mathcal{C}_1|} - \dfrac{b(\mathcal{C})}{|\mathcal{C}|} \geq \frac{1}{n(n-1)} > \frac{1}{n^2}.
\end{align*}
Then, we have
\begin{align*}
	\frac{c(t_0;\mathcal{C}_1)}{|\mathcal{C}_1|} - \frac{c(t_0;\mathcal{C})}{|\mathcal{C}|}
	&= \left(\frac{a(\mathcal{C}_1)}{|\mathcal{C}_1|} - \frac{a(\mathcal{C})}{|\mathcal{C}|}\right) -\left(\frac{b(\mathcal{C}_1)}{|\mathcal{C}_1|} - \frac{b(\mathcal{C})}{|\mathcal{C}|} \right) t_0\\
	&\geq (-2\max_{e\in E}|a(e)|) + \left(\frac{b(\mathcal{C}_1)}{|\mathcal{C}_1|} - \frac{b(\mathcal{C})}{|\mathcal{C}|} \right) \cdot (2n^2\max_{e\in E}|a(e)| +1)\\
	&>  (-2\max_{e\in E}|a(e)|) + \frac{1}{n^2} \cdot 2n^2\max_{e\in E}|a(e)| \\
	&= 0.
\end{align*}
This contradicts the fact that $\mathcal{C}_1$ is the minimum mean cycle with respect to $c(t_0)$.
\endproof

From now on, suppose that $N$, $t_0$, and $t_k, \alpha_k, \beta_k, \gamma_k, \mathcal{C}_k$ for $k \in [N]$ are obtained during Algorithm~\ref{alg:main}.
The next lemma shows that Algorithm~\ref{alg:main} maintains $\mathcal{T}$ to be the $r$-rooted shortest path tree when the cost function is modified.

\begin{lemma}		\label{lem:costbp}
For $k \in [N-1]$, we have $c_{\gamma_k}(t_k;e) = c_{\gamma_{k+1}}(t_k;e)$ for all $e \in E$.
\end{lemma}

\proof
When the if-block in lines 9--16 of Algorithm~\ref{alg:main} is executed, we have 
\begin{align*}
	c_{\gamma_{k}}(t_k;\mathcal{C}_{k+1}) = (a(\mathcal{C}_{k+1}) - b(\mathcal{C}_{k+1}) t_k) - (\alpha_{k}-\beta_{k}t_k) |\mathcal{C}_{k+1}| = 0.
\end{align*}
According to the substitution in line 12, we have $a(\mathcal{C}_{k+1}) = \alpha_{k+1}|\mathcal{C}_{k+1}|$ and $b(\mathcal{C}_{k+1}) = \beta_{k+1}|\mathcal{C}_{k+1}|$.
Hence, we obtain
\begin{align*}
	(\alpha_{k+1} - \beta_{k+1} t_k) |\mathcal{C}_{k+1}|   - (\alpha_{k}-\beta_{k}t_k) |\mathcal{C}_{k+1}| = 0,
\end{align*}
yielding
\begin{align*}
	\alpha_{k+1} - \beta_{k+1} t_k =  \alpha_{k}-\beta_{k}t_k.
\end{align*}
Thus, we have $c_{\gamma_k}(t_k;e) = c_{\gamma_{k+1}}(t_k;e)$ for all $e \in E$.
\endproof

We define the globally modified cost function $c_*: E \to \mathbb{R}$ by
\begin{align*}
	c_*(t;e) := c_{\gamma_k}(t;e) = (a(e)-b(e)t) - (\alpha_k - \beta_k t) \quad \text{if } t \in [t_{k-1},t_k],
\end{align*}
which is continuous and piecewise linear as a function of $t \in \mathbb{R}$.
Note that $c_*(t;e)$ is well-defined at breakpoints because of Lemma~\ref{lem:costbp}. The case where $t < t_0$ is confirmed by Lemma~\ref{lem:initial}.
At any iteration of Algorithm~\ref{alg:main}, $\mathcal{T}_v$ is the shortest $v$-$r$ path with respect to $c_*(t)$.

\subparagraph{}
Next, we estimate an upper bound for the number of breakpoints.

\begin{lemma}		\label{lem:samebl}
Suppose that two paths $\mathcal{P}_1, \mathcal{P}_2$ satisfy $a(\mathcal{P}_1) \leq a(\mathcal{P}_2)$, $b(\mathcal{P}_1) = b(\mathcal{P}_2)$, and $|\mathcal{P}_1| = |\mathcal{P}_2|$. 
Then, we have $c(t;\mathcal{P}_1) \leq c(t;\mathcal{P}_2)$ and $c_*(t;\mathcal{P}_1) \leq c_*(t;\mathcal{P}_2)$
for any $t \in \mathbb{R}$.
\end{lemma}

\proof
Since $c(t;\mathcal{P}_i) = a(\mathcal{P}_i) - b(\mathcal{P}_i) t$ for $i=1,2$, we have $c(t;\mathcal{P}_1) \leq c(t;\mathcal{P}_2)$.
In addition, since 
\begin{align*}
	c_*(t;\mathcal{P}_i) = a(\mathcal{P}_i) - b(\mathcal{P}_i) t - (\alpha_k - \beta_k t) |\mathcal{P}_i| \quad \text{if } t \in [t_{k-1},t_k],
\end{align*}
for $i=1,2$, we have $c_*(t;\mathcal{P}_1) \leq c_*(t;\mathcal{P}_2)$ for all $t \in \mathbb{R}$.
\endproof

Let us consider the set $L_{\mathrm{cycle}} := \{ (|\mathcal{C}|, b(\mathcal{C})) \mid \mathcal{C} \text{ is a cycle} \}$.
For a cycle $\mathcal{C}$ with length $\ell$, we observe that $b(\mathcal{C})$ can take at most $2\ell W+1$ different values, say, $0, \pm 1, \dots, \pm \ell W$.
Hence, we have
\begin{align}
	|L_{\mathrm{cycle}}| = \sum_{\ell=1}^n (2\ell W+1) \leq (n^2+2n)W.	\label{eq:Lbound}
\end{align}

\begin{proposition}		\label{prop:mmcchange}
There exist a positive integer $M \leq |L_{\mathrm{cycle}}|$ and cycles $\mathcal{C}^{(1)},\mathcal{C}^{(2)},\dots,\mathcal{C}^{(M)}$ such that
\begin{align*}
	\gamma(t) = \min_{i=1,2,\dots,M} \frac{c(t;\mathcal{C}^{(i)})}{|\mathcal{C}^{(i)}|}.
\end{align*}
In particular, $\gamma(t)$ is the minimum of at most $(n^2+2n)W$ linear polynomials 
\begin{align*}
	\gamma_i := \frac{a(\mathcal{C}^{(i)})}{|\mathcal{C}^{(i)}|} - \frac{b(\mathcal{C}^{(i)})}{|\mathcal{C}^{(i)}|}\, t, \quad i=1,2,\dots,M.
\end{align*}
\end{proposition}

\proof
Let $L_{\mathrm{cycle}} = \{(\ell^{(i)}, b^{(i)}) \mid i=1,2,\dots,|L_{\mathrm{cycle}}|\}$.
For each $(\ell^{(i)}, b^{(i)}) \in L_{\mathrm{cycle}}$, we take a cycle $\mathcal{C}^{(i)}$ that minimizes $a(\mathcal{C})$ among all cycles $\mathcal{C}$ with $|\mathcal{C}| = \ell^{(i)}$ and $b(\mathcal{C}) = b^{(i)}$.
By Lemma~\ref{lem:samebl}, $\mathcal{C}^{(i)}$ also minimizes $c(t;\mathcal{C})$ for all $t \in \mathbb{R}$ among all those cycles.
Hence, by removing redundant cycles and renumbering, the minimum cycle mean $\gamma(t)$ is attained by some of $\mathcal{C}^{(1)},\mathcal{C}^{(2)},\dots,\mathcal{C}^{(M)}$ for any $t \in \mathbb{R}$, where $M \leq |L_{\mathrm{cycle}}|$.
\endproof

To evaluate the overall time complexity of the algorithm, we need to estimate the upper bound for the number of changes in the shortest path $\mathcal{T}_v$ for each $v \in V$.
For all $v \in V$, let us define the set $L_v := \{ (|\mathcal{P}|, b(\mathcal{P})) \mid \mathcal{P} \text{ is a $v$-$r$ path} \}$.
As in the case for $L_{\mathrm{cycle}}$, we have
\begin{align}
	|L_v| = \sum_{\ell=0}^{n-1} (2\ell W+1) \leq n^2W.	\label{eq:Lvbound}
\end{align}
Hence, $b_{\beta_k}(\mathcal{P}) = - b(\mathcal{P}) + \beta_k |\mathcal{P}|$ for $v$-$r$ paths $\mathcal{P}$ may take at most $n^2W$ different values for each $k$.
An easy extension from the parametric shortest path algorithm is that $\mathcal{T}_v$ changes at most $n^2W$ times between $t_k$ and $t_{k+1}$, which causes at most $O(n^4W^2)$ changes in total.
However, we can improve this estimate.
The following observation is useful for this purpose.

\begin{lemma}		\label{lem:swap}
Let $p_1,p_2,\dots,p_M$ be $M$ letters.
Starting from the sequence $p_1p_2\cdots p_M$, we swap two adjacent letters repeatedly.
Assume that each pair of letters $p_i$ and $p_j$ is swapped at most twice.
Then, the letter occupying the first position in the sequence changes at most $2(M-1)$ times.
\end{lemma}

\proof
We prove the lemma by induction on $M$.
The case where $M=1$ is trivial.
Suppose that the assertion is true up to $M$.
Let us consider $M+1$ letters $p_0,p_1,\dots,p_M$ and the initial sequence $p_0p_1\cdots p_M$.
For $t \in \mathbb{Z}_{> 0}$, we denote $p_i \prec_t p_j$ if $p_i$ precedes $p_j$ in the sequence just after the $t$th swap. 
If some $i \in [M]$ satisfies $p_0 \prec_t p_i$ for all $t$, then the letter $p_i$ will never occupy the first position of the sequences.
Hence, we may ignore $p_i$ and consider only the remaining $M$ letters.
By induction, the first letter changes at most $2(M-1) \leq 2((M+1)-1)$ times.

Next, we assume that for all $i \in [M]$ there exists $t \in \mathbb{Z}_{>0}$ such that $p_i \prec_t p_0$.
By the assumption of the lemma, suppose that the swap of $p_0$ and $p_i$ occurs at the $t^-_i$th and $t^+_i$th swaps overall.
If the swap of $p_0$ and $p_i$ occurs exactly once, we set $t^+_i = \infty$.
Then, we observe that $p_i \prec_t p_0$ if $t \in [t^-_i, t^+_i) \cap \mathbb{Z}$ and $p_0 \prec_t p_i$ otherwise.
Let $t^- = \min_{i \in [M]} t^-_i$ and $t^+ = \max\{ t \mid [t^-,t) \subset \bigcup_{i=1}^M [t_i^-,t_i^+)\}$.
Subsequently, let $I = \{i \in [M] \mid [t^-_i, t^+_i) \subset [t^-,t^+)\}$ and $J = [M] \setminus I$.
By the maximality of $t^+$, we have $[t^-_j, t^+_j) \cap [t^-,t^+) = \emptyset$ for all $j \in J$,
yielding $p_i \prec_t p_0 \prec_t p_j$ for all $i \in I, j \in J$, and $t \in [t^-,t^+) \cap \mathbb{Z}$.
Hence, for any $t \in [t^-,t^+) \cap \mathbb{Z}$, the minimum letter with respect to $\prec_t$ is $p_i$ for some $i \in I$.
By induction, the first letter of the sequence changes at most $2(|I|-1)$ times between the $t^-$th and the $t^+$th swaps.
On the other hand, for any $t \in [t^+,\infty) \cap \mathbb{Z}$, the minimum letter with respect to $\prec_t$ is $p_j$ for some $j \in J \cup \{0\}$.
By induction, the first letter changes at most $2((|J|+1)-1)$ times after the $t^+$th swap.
By taking into account the changes occurring at the $t^-$th and the $t^+$th swaps, the total number of changes in the first letter of the sequence is at most
\begin{align*}
	2(|I|-1) + 2((|J|+1)-1) + 2 = 2(|I| + |J|) = 2((M+1)-1).
\end{align*}
This proves the assertion for $M+1$.
\endproof

\begin{proposition}		\label{prop:treechange}
During Algorithm~\ref{alg:main}, the shortest path $\mathcal{T}_v$ is updated at most $2(|L_v|-1)$ times for each $v \in V$.
\end{proposition}

\proof
Let $|L_v| = M$ and $L_v = \{(\ell^{(i)}, b^{(i)}) \mid i \in [M] \}$.
For each $(\ell^{(i)}, b^{(i)}) \in L_v$, we take a $v$-$r$ path $\mathcal{P}_i$ that minimizes $a(\mathcal{P})$ among $v$-$r$ paths with $|\mathcal{P}| = \ell^{(i)}$ and $b(\mathcal{P}) = b^{(i)}$.
By Lemma~\ref{lem:samebl}, $\mathcal{P}_i$ also minimizes $c_*(t;\mathcal{P})$ for all $t \in \mathbb{R}$ among all those paths.
Using this path $\mathcal{P}_i$, we define functions $p_i(t)$ and $d_i(t)$ on $\mathbb{R}$ by
\begin{alignat*}{2}
	p_i(t) &:= c_*(t;\mathcal{P}_i) = (a(\mathcal{P}_i)-b(\mathcal{P}_i)t) - (\alpha_k - \beta_k t) |\mathcal{P}_i| &\quad& \text{if } t \in [t_{k-1},t_k], \\
	d_i(t) &:= -b_{\beta_k}(\mathcal{P}_i) = - b(\mathcal{P}_i) + \beta_k |\mathcal{P}_i| &\quad& \text{if } t \in [t_{k-1},t_k).
\end{alignat*} 
Then, $d_i(t)$ is the right-derivative of $p_i(t)$.

Now, we regard $p_1,p_2,\dots,p_{M}$ as $M$ letters and consider sequences of them.
We write $p_i \prec p_j$ if the letter $p_i$ precedes $p_j$ in the sequence under consideration.
We introduce the swap procedure as follows:
\begin{enumerate}
\item Arrange the initial sequence $p_1p_2\cdots p_{M}$ so that $p_1(t_0) \leq p_2(t_0) \leq \cdots \leq p_{M}(t_0)$, where $t_0 = -2n^2\max_{e\in E} |a(e)|-1$.
\item Set $t^* := t_0$.
\item If there exist adjacent letters $p_i \prec p_j$ such that $p_i(t^*)=p_j(t^*)$ and $d_i(t^*) < d_j(t^*)$, then swap $p_i$ and $p_j$.
Continue this process until no such pairs remain.
\item If there exist adjacent letters $p_i \prec p_j$ such that $p_i(t^*)=p_j(t^*)$ and $d_i(t^*) > d_j(t^*)$, then swap $p_i$ and $p_j$.
If there are multiple such pairs, we perform the swap of the pair that is minimal with respect to $\prec$.
Continue this process until no such pairs remain.
\item Take the minimum value $t > t^*$ satisfying the following condition:
\begin{itemize}
\item There exist distinct letters $p_i$ and $p_j$ such that $p_i(t) = p_j(t)$ and $d_i(t) \neq d_j(t)$.
\end{itemize}
If there is no such $t$, terminate the procedure.
\item Set $t^* := t$ and back to step 3. 
\end{enumerate}

\renewcommand{\qedsymbol}{$\blacksquare$}
\begin{claim}
During the above procedure, the swap of $p_i$ and $p_j$ occurs at most twice for each pair $(p_i,p_j)$.
\end{claim}
\proof[Proof of Claim]
We may assume without loss of generality that $|\mathcal{P}_i| \geq |\mathcal{P}_j|$.
We first note that if a swap occurs in step 3, then the swap of the same pair also occurs in step 4.
Indeed, among the letters $p_i$ with the same values of $p_i(t^*)$, they are arranged in non-increasing order of $d_i(t^*)$ at the end of step 3 and then in non-decreasing order of $d_i(t^*)$ at the end of step 4.

Suppose that the swap of $p_i$ and $p_j$ occurs in step 4 at $t^*$.
If $p_j \prec p_i$ at the end of step 4, we have $d_j(t^*) < d_i(t^*)$.
We observe that
\begin{align*}
	d_i(t) - d_j(t) = - b(\mathcal{P}_i) + b(\mathcal{P}_j) + \beta_k (|\mathcal{P}_i| - |\mathcal{P}_j|) \quad  \text{if } t \in [t_{k-1},t_k).
\end{align*}
Since $\beta_k$ increases as $k$ increases, $d_i(t) - d_j(t)$ is non-decreasing as $t$ increases.
Hence, we have $d_j(t) < d_i(t)$ for all $t > t^*$.
Since $p_i(t^*)=p_j(t^*)$, we have $p_i(t) > p_j(t)$ for all $t > t^*$.
This implies that the swap of $p_i$ and $p_j$ will never occur at $t > t^*$.
By taking into account the earlier swap that causes $p_i \prec p_j$, the swap of $p_i$ and $p_j$ occurs at most twice.
In particular, the above procedure terminates after a finite number of swaps.
\endproof
\begin{claim}
The path $\mathcal{P}_i$ corresponding to the first letter $p_i$ in the sequence at $t^*$ is the shortest $v$-$r$ path with respect to $c_*(t^*)$.
\end{claim}
\proof[Proof of Claim]
If $p_i(t^*)=p_j(t^*)$ and $p_j \prec p_i$ at the end of step 4, then we have $p_j(t) \leq p_i(t)$ for all $t \in [t^*, t']$, where $t'$ is the minimum value of $t > t^*$ satisfying the condition in step 5.
This is only because of $d_j(t) \leq d_i(t)$ and does not depend on whether $|\mathcal{P}_i| \geq |\mathcal{P}_j|$ or $|\mathcal{P}_i| < |\mathcal{P}_j|$.
In particular, $p_j \prec p_i$ leads to $p_j(t') \leq p_i(t')$.
Hence, we inductively observe that the letters in the sequence are arranged in non-decreasing order of $p_{i}(t^*) = c_*(t;\mathcal{P}_i)$ at any iteration of the procedure.
\endproof

\renewcommand{\qedsymbol}{$\square$}
\proof[Proof of Proposition~\ref{prop:treechange} (cont.)]
When the path $\mathcal{T}_v$ is updated at $t=\min_{v \in V} \kappa(v)$, the old path $\mathcal{T}_{v,\mathrm{old}}$ and the new one $\mathcal{T}_{v,\mathrm{new}}$ satisfy $c_*(t; \mathcal{T}_{v,\mathrm{old}}) = c_*(t; \mathcal{T}_{v,\mathrm{new}})$ and $b_{\beta_k}(\mathcal{T}_{v,\mathrm{old}}) < b_{\beta_k}(\mathcal{T}_{v,\mathrm{new}})$.
This implies that there exist $v$-$r$ paths $\mathcal{P}_i$ and $\mathcal{P}_j$ such that $p_i(t) = p_j(t)$ and $d_i(t) > d_j(t)$.
Hence, this value $t$ must have been detected at step 5 in the above procedure. 
If several updates of $\mathcal{T}_v$ occur at the same value $t$, each path $\mathcal{T}_v$ corresponds to $\mathcal{P}_i$ with the same value of $p_i(t)$ but distinct value of $d_i(t)$.
By noting the ordering established in step 3 and the priority for swaps in step 4, all corresponding letters $p_i$ have come to the first position of the sequence during step 4.
This implies that the number of times $\mathcal{T}_v$ is updated is at most the number of times the first letter has changed throughout the procedure.
By Lemma~\ref{lem:swap}, the letter occupying the first position changes at most $2(M-1)$ times in total.
Hence, $\mathcal{T}_v$ is updated at most $2(|L_v|-1)$ times for each $v \in V$.
\endproof

\begin{theorem}	\label{thm:main}
Algorithm~\ref{alg:main} solves the parametric MCM in $O((m+n \log n)n^2W)$ time.
\end{theorem}

\proof
By Proposition~\ref{prop:psp} and Lemma~\ref{lem:costbp}, $\mathcal{T}$ is the $r$-rooted shortest path tree throughout Algorithm~\ref{alg:main}.
For $k \in [N]$, we have $c_{\gamma_k}(t;\mathcal{C}_k) = 0$ for all $t \in \mathbb{R}$.
The iteration based on the parametric shortest path algorithm ensures that there are no negative-cost cycles with respect to $c_{\gamma_k}(t)$ for any $t \in [t_{k-1},t_k]$.
With respect to $c(t)$ and $c_{\gamma_k}(t)$, the cycle means of all cycles equally differ by $\gamma_k$. 
Hence, for all $t \in [t_{k-1},t_k]$, the minimum cycle mean with respect to $c(t)$ is $\gamma_k = \alpha_k - \beta_k t$, which is attained by $\mathcal{C}_k$.
This implies that Algorithm~\ref{alg:main} solves the parametric MCM correctly.

Next, we prove the complexity result.
The initialization in lines 1--6 can be done using Karp's method~\cite{Karp1978} and Bellman-Ford's method~\cite{Bellman1958}, which requires $O(mn)$ time in total.
By Proposition~\ref{prop:mmcchange}, the if-block starting with line 9 is executed at most $(|L_{\mathrm{cycle}}|-1)$ times.
Since $\sum_{v\in V} \deg^+(v) = m$, each iteration of lines 10--16 can be done in $O(m)$ time.
Hence, by~\eqref{eq:Lbound}, the if-block in lines 9--16 requires $O(mn^2W)$ time overall Algorithm~\ref{alg:main}.

Let us assume that the vertex keys $\kappa(v)$ for $v \in V$ are maintained by using the Fibonacci heap data structure.
By Proposition~\ref{prop:treechange}, the path $\mathcal{T}_v$ is updated at most $2(|L_v|-1)$ times for each $v \in V$.
Hence, the computation of the key values in line 24 can be done in $O(mn^2W)$ time in total because
we have $\sum_{v \in V} \deg^+(v) \cdot |L_v| \leq mn^2W$ using the inequality~\eqref{eq:Lvbound}.
The search for the minimum key value in line 7 is performed at most $|L_{\mathrm{cycle}}| + \sum_{v \in V} 2(|L_v|-1)$ times.
The key value increment in line 24 is executed at most $\sum_{v \in V} 2(|L_v|-1)$ times in total.
The amortized time for these operations is $O(n^3 W \log n)$.
The decrement of the key value in line 28 is executed at most $\deg^-(v) \cdot 2(|L_v|-1)$ times for each $v \in V$, and hence at most $2mn^2W$ times in total.
Hence, the amortized time for the key value decrement is $O(mn^2W)$.
Thus, the time complexity to maintain the key values by the Fibonacci heap data structure is $O((m+n\log n)n^2W)$ in total.
In summary, the overall time complexity of Algorithm~\ref{alg:main} is $O((m+n \log n)n^2W)$.
\endproof

\begin{corollary}
If $b(e) \in \{-1,0,1\}$ for all $e \in E$, then Algorithm~\ref{alg:main} solves the parametric MCM in $O((m+n \log n)n^2)$ time.
\end{corollary}

The following refined complexity is obtained within the proof of Theorem~\ref{thm:main} and would be useful when only a few edges are parametric. 

\begin{corollary}	\label{coro:complexity}
Algorithm~\ref{alg:main} solves the parametric MCM in $O(mn+K)$ time, where
\begin{align*}
	K = m \cdot |L_{\mathrm{cycle}}| + \sum_{v\in V} (\deg^+(v) + \deg^-(v) + \log n) \cdot |L_v|.
\end{align*}
\end{corollary}

\section{Parametric tropical eigenvectors}

The tropical semiring, which is also called the min-plus algebra, is the semiring of numbers $\mathbb{R}_{\min} := \mathbb{R} \cup \{\infty\}$ with addition $\oplus$ and multiplication $\otimes$ defined by
\begin{align*}
	a \oplus b := \min(a,b), \quad a \otimes b := a+b, \quad a,b \in \mathbb{R}_{\min}.
\end{align*}
The neutral element for the addition and the unit element for the multiplication are $\infty$ and $0$, respectively.
The tropical semiring satisfies the same axioms as a ring, except for the existence of additive inverses.
For $a \in \mathbb{R}_{\min}$ and $k \in \mathbb{Z}$, the tropical power is defined by $a^{\otimes k} := ka$.
For $m,n \in \mathbb{Z}_{>0}$, let $\mathbb{R}_{\min}^{m \times n}$ be the set of all $m$-by-$n$ matrices over $\mathbb{R}_{\min}$.
In particular, the set of $n$-dimensional column vectors is $\mathbb{R}_{\min}^n := \mathbb{R}_{\min}^{n \times 1}$.
Arithmetic operations on tropical matrices are defined analogously to those in conventional matrix algebra.
The sum of $A = (a_{ij}) \in \mathbb{R}_{\min}^{m \times n}$ and $B = (b_{ij}) \in \mathbb{R}_{\min}^{m \times n}$ is defined as
\begin{align*}
	A \oplus B := (a_{ij} \oplus b_{ij}) \in \mathbb{R}_{\min}^{m \times n}.
\end{align*}
The product of $A = (a_{ij}) \in \mathbb{R}_{\min}^{\ell \times m}$ and $B = (b_{ij}) \in \mathbb{R}_{\min}^{m \times n}$ is defined as
\begin{align*}
	A \otimes B := \left(\bigoplus_{k=1}^m a_{ik} \otimes b_{kj} \right) \in \mathbb{R}_{\min}^{\ell \times n}.
\end{align*}
The scalar multiplication of $A = (a_{ij}) \in \mathbb{R}_{\min}^{m \times n}$ and $\lambda \in \mathbb{R}_{\min}$ is defined as
\begin{align*}
	\lambda \otimes A := (\lambda \otimes a_{ij}) \in \mathbb{R}_{\min}^{m \times n}.
\end{align*}
The zero matrix (or vector) is the matrix (or vector) whose entries are all $\infty$.
The identity matrix, denoted by $I$, is the square matrix with $0$ on the diagonal and $\infty$ elsewhere.
For $A \in \mathbb{R}_{\min}^{n \times n}$ and $k \in \mathbb{Z}_{>0}$, the tropical power $A^{\otimes k}$ means the $k$-fold product of $A$.

The associated digraph of $A \in \mathbb{R}_{\min}^{n \times n}$ is the weighted digraph $\mathcal{G}(A) := (V,E,c)$, where $V = [n]$, $E = \{ (i,j) \mid a_{ij} \neq \infty\}$, and $c: E \to \mathbb{R}$ is the cost function defined by $c((i,j)) = a_{ij}$ for $(i,j) \in E$.
For $k \in \mathbb{Z}_{>0}$, we readily observe that the $(i,j)$ entry of $A^{\otimes k}$ is equal to the minimum cost of all $i$-$j$ walks with length $k$. 
If $\mathcal{G}(A)$ has no negative-cost cycle, then the $(i,j)$ entry of $A^* := I \oplus A \oplus A^{\otimes 2} \oplus \cdots \oplus A^{\otimes (n-1)}$ is equal to the cost of the shortest $i$-$j$ path.

For $A \in \mathbb{R}_{\min}^{n \times n}$, if there exist $\lambda \in \mathbb{R}_{\min}$ and $x \in \mathbb{R}_{\min}^n \setminus \{(\infty,\dots,\infty)^\top\}$ such that 
\begin{align*}
	A \otimes x = \lambda \otimes x, 
\end{align*}
then $\lambda$ is called an eigenvalue of $A$ and $x$ is called an eigenvector of $A$ with respect to $\lambda$.
The eigenproblem in the tropical semiring can be understood in terms of the MCM in $\mathcal{G}(A)$.

\begin{proposition}[\cite{Green1979}]		\label{prop:tropeig}
The minimum eigenvalue of $A \in \mathbb{R}_{\min}^{n \times n}$ is equal to the minimum cycle mean of $\mathcal{G}(A)$.
If $\mathcal{G}(A)$ is strongly connected, then the eigenvalue of $A$ is unique.
For the minimum eigenvalue $\lambda$, the $j$th column of $((-\lambda) \otimes A)^*$ is an eigenvector of $A$ with respect to $\lambda$ if and only if the vertex $j$ is on the minimum mean cycle.
\end{proposition}

\subsection{Parametric algorithm}

Let us consider a linear parametric matrix $C(t) = (c_{ij}(t)) \in \mathbb{R}_{\min}^{n \times n}$, where $c_{ij}(t) = a_{ij} - b_{ij}t$ with $a_{ij},b_{ij} \in \mathbb{Z}$ or $c_{ij}(t) = \infty$.
Let $m = |\{ (i,j) \mid c_{ij}(t) \neq \infty \}|$ and $W = \max_{b_{ij} \neq \infty} |b_{ij}|$.
The following result is straightforward from Theorem~\ref{thm:main}.

\begin{corollary}
The eigenvalue of $C(t)$ can be computed as a piecewise-linear function of $t$ in $O((m+n \log n)n^2W)$ time.
\end{corollary}

To compute the corresponding eigenvectors using Algorithm~\ref{alg:main}, we need to pay additional attention to the choice of the root $r$ of the shortest path tree.
Let $t_k, \gamma_k$, and $\mathcal{C}_k$ be obtained during Algorithm~\ref{alg:main} for $\mathcal{G}(C(t))$.
As mentioned in Proposition~\ref{prop:tropeig}, the eigenvector of $C(t)$ for $t \in [t_{k-1},t_k]$ corresponding to the eigenvalue $\gamma_k = \alpha_k - \beta_k t$ is given as the $j$th column of $((-\gamma_k) \otimes C(t))^*$ if $j$ is on $\mathcal{C}_k$.
We observe that the matrix $(-\gamma_k) \otimes C(t)$ represents the cost function $c_{\gamma_k}(t)$.
Hence, the $(i,j)$ entry of $((-\gamma_k) \otimes C(t))^*$ is equal to the cost of the shortest $i$-$j$ path with respect to $c_{\gamma_k}(t)$.
Thus, to obtain the $j$th column vector, it is sufficient to have the $j$-rooted shortest path tree with respect to $c_{\gamma_k}(t)$.

With this in mind, we dynamically change the root of the shortest path tree instead of fixing the root to $r$.
We need to change the root when the minimum mean cycle is updated, that is, when the if-block in lines 9--16 is executed.
We choose any vertex $j$ on the new minimum mean cycle $\mathcal{C}_{k+1}$ and reconstruct the $j$-rooted shortest path tree with respect to $c_{\gamma_{k+1}}(t_k)$. 
Since we have already obtained the shortest path tree $\mathcal{T}$ at another root, the single-sink shortest path problem to $j$ can be solved by Dijkstra's algorithm~\cite{Dijkstra1959} using the nonnegative reduced costs induced by $\mathcal{T}$.
Hence, the time complexity for the reconstruction of the shortest path tree is $O(m + n\log n)$.

Since the column vectors of $((-\gamma_k) \otimes C(t))^*$ depend on the shortest path tree, we need to retain all values $\min_{v \in V} \kappa(v)$ as breakpoints even when the else-block in lines 17--28 is executed.
In addition, the changes in the path $\mathcal{T}_v$ are counted separately for each root $j \in V$.
Hence, the total complexity of the else-block is multiplied by $n$ in the worst case.
In summary, we have the following result.

\begin{theorem}
The eigenvector of $C(t)$ can be computed as the vector whose components are piecewise-linear functions of $t$ in $O((m+n\log n)n^3W)$ time.
\end{theorem}

\section{Comparison to the symbolic Karp's method}

To solve the parametric MCM, we may use Karp's method~\eqref{eq:mcmkarp} by regarding the parameter $t$ as a symbol representing an indeterminate.
For $\ell \in \mathbb{Z}_{\geq 0}$ and $v, r \in V$, recall that $F_v(\ell)$ is the minimum cost of all $v$-$r$ walks with length $\ell$.
In terms of the tropical semiring, $F_v(\ell)$ is a tropical Laurent polynomial of the form
\begin{align*}
	F_v(\ell) = \bigoplus_{i=-\ell W}^{\ell W} q_i^{(v,\ell)} \otimes t^{\otimes (-i)}, \quad q_i^{(v,\ell)} \in \mathbb{R}_{\min}.
\end{align*}
Note that $q_i^{(v,\ell)}$ is the minimum value of $a(\mathcal{P})$ over all $v$-$r$ walks $\mathcal{P}$ with length $\ell$ such that $b(\mathcal{P}) = i$ if one exists; otherwise $q_i^{(v,\ell)} = \infty$.
By setting $F_v(0) = 0$ for all $v \in V$, we inductively compute $F_v(\ell)$ by
\begin{align*}
	F_v(\ell+1) = \min_{e \in \delta^+(v)} (a(e) - b(e)t + F_{e^+}(\ell)) 
	=\bigoplus_{e \in \delta^+(v)} \bigoplus_{i=-\ell W}^{\ell W} a(e) \otimes q_i^{(e^+,\ell)} \otimes t^{\otimes (-i - b(e))}
\end{align*}
for $v \in V$ and $\ell = 0,1,\dots,n-1$.
The computation of $F_v(\ell+1)$ for all $v \in V$ from $F_1(\ell),F_2(\ell),\dots,F_n(\ell)$ requires $m(2\ell W+1)$ tropical multiplications $a(e) \otimes q_i^{(e^+,\ell)}$.
Hence, we can obtain $F_v(\ell)$ for all $v \in V$ and $\ell \in [n]$ in $O(mn^2W)$ time.
Even though the terms with $q_i^{(v,\ell)} = \infty$ are ignored, the computation requires $\sum_{v \in V} \deg^-(v) \cdot |L'_v|$ times of tropical multiplications, where $L'_v = \{ (|\mathcal{P}|, b(\mathcal{P})) \mid \mathcal{P} \text{ is a $v$-$r$ walk, $|\mathcal{P}| \leq n-1$} \}$.
This time complexity of the symbolic Karp's method is not significantly better than the result in Theorem~\ref{thm:main} or Corollary~\ref{coro:complexity}; the improvement is limited to the case where $m$ is much smaller than $n\log n$.

In addition, in the computation of the minimum mean cycle by~\eqref{eq:mcmkarp}, we need all terms of $F_v(\ell)$ for all $v \in V$ and $\ell \in [n]$.
This means that the symbolic Karp's method requires $O(\sum_{v \in V} \deg^-(v) \cdot |L'_v|)$ time of computation for any instance.
On the other hand, the time complexity in Corollary~\ref{coro:complexity} based on the parametric shortest path algorithm is a worst-case bound.
As noted in Section 2.2, the parametric shortest path algorithm for the non-parametric MCM is faster than Karp's method in practice.
For an instance of the parametric MCM to which a similar observation applies, the expected complexity of Algorithm~\ref{alg:main} will be much smaller, outperforming the symbolic Karp's method.

Another advantage of Algorithm~\ref{alg:main} is that it is easily modified to the tropical eigenproblem as described in Section 4.
Since the computation of tropical eigenvectors is reduced to the shortest path problem with respect to the modified cost function $c_*(t)$, the dynamic modification to the cost function in Algorithm~\ref{alg:main} is suitable for this.

\section{Concluding remarks}

In this paper, we proposed an algorithm to solve the parametric MCM.
Although the time complexity of the algorithm is $O((m+n\log n)n^2W)$ in general, it is much faster when the numbers of breakpoints for the minimum mean cycles and/or shortest path trees are small.
This is considered to be an advantage over the symbolic Karp's method.
Our algorithm can be translated into the setting of the tropical spectral theory.
By using shortest path trees, we can compute the tropical eigenvectors of a parametric matrix as a byproduct of computing the eigenvalues.
Although at present, the worst-case time complexity to compute tropical eigenvectors of a parametric matrix is $O((m+n\log n)n^3W)$, this bound may be improved if partial optimality of shortest paths can be effectively exploited.
This remains a topic for future work.

\section*{Acknowledgments}
This work is supported by JSPS KAKENHI Grant No.~26K06897.

\bibliographystyle{abbrv}
\bibliography{paraMCMref_arxiv}

@article{Banerjee2003,
  author       = {Ipsita Banerjee and Marianthi G. Ierapetritou},
  title        = {Parametric process synthesis for general nonlinear models},
  journal      = {Computers and Chemical Engineering },
  volume       = {27},
  number       = {10},
  pages        = {1499--1512},
  year         = {2003},
  doi          = {10.1016/S0098-1354(03)00096-6}
}

@article{Bellman1958,
 author       = {Richard Bellman},
 journal      = {Quarterly of Applied Mathematics},
 number       = {1},
 pages        = {87--90},
 publisher    = {Brown University},
 title        = {ON A ROUTING PROBLEM},
 volume       = {16},
 year         = {1958}
}

@article{Cechlarova2005,
  author       = {Katar{\'{\i}}na Cechl{\'{a}}rov{\'{a}}},
  title        = {Eigenvectors of interval matrices over max-plus algebra},
  journal      = {Discrete Applied Mathematics},
  volume       = {150},
  number       = {1-3},
  pages        = {2--15},
  year         = {2005},
  doi          = {10.1016/J.DAM.2005.02.016}
}

@Article{Cohen1985,
  author={Guy Cohen and Didier Dubois and Jean{-}Pierre Quadrat and Michel Viot},
  journal={IEEE Transactions on Automatic Control}, 
  title={A linear-system-theoretic view of discrete-event processes and its use for performance evaluation in manufacturing}, 
  year={1985},
  volume={30},
  number={3},
  pages={210--220},
  doi={10.1109/TAC.1985.1103925},
}

@article{Dasdan2004,
  author       = {Ali Dasdan},
  title        = {Experimental analysis of the fastest optimum cycle ratio and mean algorithms},
  journal      = {{ACM} Transactions on Design Automation of Electronic Systems},
  volume       = {9},
  number       = {4},
  pages        = {385--418},
  year         = {2004},
  doi          = {10.1145/1027084.1027085}
}

@article{Dijkstra1959,
  author       = {Edsger W. Dijkstra},
  title        = {A note on two problems in connexion with graphs},
  journal      = {Numerische Mathematik},
  volume       = {1},
  pages        = {269--271},
  year         = {1959},
  doi          = {10.1007/BF01386390}
}

@article{Eisner1976,
  author       = {Mark J. Eisner and
                  Dennis G. Severance},
  title        = {Mathematical Techniques for Efficient Record Segmentation in Large Shared Databases},
  journal      = {Journal of the {ACM}},
  volume       = {23},
  number       = {4},
  pages        = {619--635},
  year         = {1976},
  doi          = {10.1145/321978.321982}
}

@article{Fredman1987,
  author       = {Michael L. Fredman and
                  Robert Endre Tarjan},
  title        = {Fibonacci heaps and their uses in improved network optimization algorithms},
  journal      = {Journal of the {ACM}},
  volume       = {34},
  number       = {3},
  pages        = {596--615},
  year         = {1987},
  doi          = {10.1145/28869.28874},
}

@article{Gaubert1998,
  title        = {The duality theorem for min-max functions},
  journal      = {Comptes Rendus de l'Acad\'{e}mie des Sciences - Series I - Mathematics},
  volume       = {326},
  number       = {1},
  pages        = {43--48},
  year         = {1998},
  doi          = {10.1016/S0764-4442(97)82710-3},
  author       = {St\'{e}phane Gaubert and Jeremy Gunawardena} 
}

@article{Gassner2010,
  author       = {Elisabeth Gassner and Bettina Klinz},
  title        = {A fast parametric assignment algorithm with applications in max-algebra},
  journal      = {Networks},
  volume       = {55},
  number       = {2},
  pages        = {61--77},
  year         = {2010},
  doi          = {10.1002/NET.20288}
}

@article{Gaubert2012,
  author       = {St{\'{e}}phane Gaubert and
                  Ricardo Katz and
                  Serge\u{\i} Sergeev},
  title        = {Tropical linear-fractional programming and parametric mean payoff
                  games},
  journal      = {Journal of Symbolic Computation},
  volume       = {47},
  number       = {12},
  pages        = {1447--1478},
  year         = {2012},
  doi          = {10.1016/J.JSC.2011.12.049}
}

@article{Gaubert2013,
  author       = {St{\'{e}}phane Gaubert and
                  Serge\u{\i} Sergeev},
  title        = {The level set method for the two-sided max-plus eigenproblem},
  journal      = {Discrete Event Dynamic Systems},
  volume       = {23},
  number       = {2},
  pages        = {105--134},
  year         = {2013},
  doi          = {10.1007/S10626-012-0137-Z}
}

@article{Gusfield1983,
  author       = {Dan Gusfield},
  title        = {Parametric Combinatorial Computing and a Problem of Program Module Distribution},
  journal      = {Journal of the {ACM}},
  volume       = {30},
  number       = {3},
  pages        = {551--563},
  year         = {1983},
  doi          = {10.1145/2402.322391}
}

@article{Karp1978,
  author       = {Richard M. Karp},
  title        = {A characterization of the minimum cycle mean in a digraph},
  journal      = {Discrete Mathematics},
  volume       = {23},
  number       = {3},
  pages        = {309--311},
  year         = {1978},
  doi          = {10.1016/0012-365X(78)90011-0}
}

@article{Karp1981,
  author       = {Richard M. Karp and
                  James B. Orlin},
  title        = {Parametric shortest path algorithms with an application to cyclic
                  staffing},
  journal      = {Discrete Applied Mathematics},
  volume       = {3},
  number       = {1},
  pages        = {37--45},
  year         = {1981},
  doi          = {10.1016/0166-218X(81)90026-3},
}

@article{Lee1987,
  author       = {Edward A. Lee and David G. Messerschmitt},
  title        = {Synchronous data flow},
  journal      = {Proceedings of the {IEEE}},
  volume       = {75},
  number       = {9},
  pages        = {1235--1245},
  year         = {1987},
  doi          = {10.1109/PROC.1987.13876}
}

@article{Megiddo1979,
  author       = {Nimrod Megiddo},
  title        = {Combinatorial Optimization with Rational Objective Functions},
  journal      = {Mathematics of Operations Research},
  volume       = {4},
  number       = {4},
  pages        = {414--424},
  year         = {1979},
  doi          = {10.1287/MOOR.4.4.414}
}

@article{Nemesch2025,
  author       = {Levin Nemesch and Stefan Ruzika and Clemens Thielen and Alina Wittmann},
  title        = {A survey of exact and approximation algorithms for linear-parametric
                  optimization problems},
  journal      = {Journal of Global Optimization},
  volume       = {93},
  number       = {1},
  pages        = {299--333},
  year         = {2025},
  doi          = {10.1007/S10898-025-01512-6}
}

@article{Nishida2026,
  author       = {Yuki Nishida and
                  Sennosuke Watanabe and
                  Yoshihide Watanabe},
  title        = {Characterization and algorithm for max-plus supereigenvector problem
                  by parametric programming},
  journal      = {Discrete Applied Mathematics},
  volume       = {379},
  pages        = {872--894},
  year         = {2026},
  doi          = {10.1016/J.DAM.2025.10.059}
}

@article{Parsons2023,
  author = {Jamie Parsons and Serge\u{\i} Sergeev and Huili Wang},
  title = {Tropical pseudolinear and pseudoquadratic optimization as parametric mean-payoff games},
  journal = {Optimization},
  volume = {72},
  number = {11},
  pages     =	 {2793--2822},
  year = {2023},
  doi = {10.1080/02331934.2022.2085100}
}

@article{Pistikopoulos2012,
  author       = {Efstratios N. Pistikopoulos and Luis F. Dom{\'{\i}}nguez and Christos Panos and Konstantinos I. Kouramas and Altannar Chinchuluun},
  title        = {Theoretical and algorithmic advances in multi-parametric programming and control},
  journal      = {Computational Management Science},
  volume       = {9},
  number       = {2},
  pages        = {183--203},
  year         = {2012},
  doi          = {10.1007/S10287-012-0144-4}
}

@article{Plavka2005,
  author       = {J{\'{a}}n Pl{\'{a}}vka},
  title        = {\emph{l}-Parametric eigenproblem in max-algebra},
  journal      = {Discrete Applied Mathematics},
  volume       = {150},
  number       = {1-3},
  pages        = {16--28},
  year         = {2005},
  doi          = {10.1016/J.DAM.2005.02.017}
}

@article{Raith2017,
  author       = {Andrea Raith and Antonio Sede{\~{n}}o{-}Noda},
  title        = {Finding extreme supported solutions of biobjective network flow problems: An enhanced parametric programming approach},
  journal      = {Computers \& Operations Research},
  volume       = {82},
  pages        = {153--166},
  year         = {2017},
  doi          = {10.1016/J.COR.2017.01.004}
}

@article{Skelin2017,
  author       = {Mladen Skelin and
                  Marc Geilen and
                  Francky Catthoor and
                  Sverre Hendseth},
  title        = {Worst-case performance analysis of {SDF}-based parameterized dataflow},
  journal      = {Microprocess. Microsystems},
  volume       = {52},
  pages        = {439--460},
  year         = {2017},
  doi          = {10.1016/J.MICPRO.2016.12.004}
}

@article{Yin2022,
  author       = {Yingxuan Yin and
                  Yuegang Tao and
                  Cailu Wang and
                  Haiyong Chen},
  title        = {Local and global robustness with $q$-step delay for max-plus linear systems},
  journal      = {Discrete Event Dynamic Systems},
  volume       = {32},
  number       = {2},
  pages        = {231--251},
  year         = {2022},
  doi          = {10.1007/S10626-021-00352-2}
}

@article{Young1991,
  author       = {Neal E. Young and
                  Robert Endre Tarjan and
                  James B. Orlin},
  title        = {Faster parametric shortest path and minimum-balance algorithms},
  journal      = {Networks},
  volume       = {21},
  number       = {2},
  pages        = {205--221},
  year         = {1991},
  doi          = {10.1002/NET.3230210206}
}

@book{Baccelli1992,
  author		= {Fran\c{c}ois Baccelli and Guy Cohen and Geert Jan Olsder and Jean{-}Pierre Quadrat},
  title			= {Synchronization and Linearity},
  address		= {Chichester},
  publisher		= {Wiley},
  year			= {1992}
}

@book{Green1979,
  author		= {Raymond A. Cuninghame{-}Green},
  title			= {Minimax Algebra},
  publisher		= {Springer-Verlag},
  address		= {Berlin, Heidelberg},
  year			= {1979},
  isbn      = {978-3-540-09113-4},
  doi       = {10.1007/978-3-642-48708-8}
}

@book{Ehrgott2005,
  author       = {Matthias Ehrgott},
  title        = {Multicriteria Optimization {(2.} ed.)},
  publisher    = {Springer},
  address		   = {Berlin, Heidelberg},
  year         = {2005},
  doi          = {10.1007/3-540-27659-9},
  isbn         = {978-3-540-21398-7}
}

@book{Gal1994,
  author		= {Tomas Gal},
  title			= {Postoptimal Analyses, Parametric Programming, and Related Topics: Degeneracy, Multicriteria Decision Making, Redundancy {(2.} ed.)},
  address		= {Berlin, New York},
  publisher		= {De Gruyter},
  year			= {1994},
  doi       = {10.1515/9783110871203}
}

@book{Heidergott2005,
  author		= {Bernd Heidergott and Geert Jan Olsder and Jacob van der Woude},
  title			= {Max Plus at Work: Modeling and Analysis of Synchronized Systems: A Course on Max-plus Algebra and Its Applications},
  address		= {Princeton},
  publisher		= {Princeton University Press},
  year			= {2005},
  isbn      = {978-0691117638}
}

@inproceedings{Bhattacharya2000,
  author       = {Bishnupriya Bhattacharya and
                  Shuvra S. Bhattacharyya},
  title        = {Parameterized dataflow modeling of {DSP} systems},
  booktitle    = {{IEEE} International Conference on Acoustics, Speech, and Signal Processing.
                  {ICASSP} 2000, Istanbul, Turkey, 5-9 June, 2000},
  pages        = {3362--3365},
  publisher    = {{IEEE}},
  year         = {2000},
  doi          = {10.1109/ICASSP.2000.860121}
}

@inproceedings{Ghamarian2008,
  author       = {Amir Hossein Ghamarian and
                  Marc Geilen and
                  Twan Basten and
                  Sander Stuijk},
  editor       = {Donatella Sciuto},
  title        = {Parametric Throughput Analysis of Synchronous Data Flow Graphs},
  booktitle    = {Design, Automation and Test in Europe, {DATE} 2008, Munich, Germany,
                  March 10-14, 2008},
  pages        = {116--121},
  publisher    = {{ACM}},
  year         = {2008},
  doi          = {10.1109/DATE.2008.4484672}
}

@inproceedings{Stoian2026,
  author       = {Mihail Stoian},
  editor       = {Meena Mahajan and
                  Florin Manea and
                  Annabelle McIver and
                  Kim Thang Nguyen},
  title        = {Mind the Gap. {D}oubling Constant Parametrization of Weighted Problems: {TSP}, Max-Cut, and More},
  booktitle    = {43rd International Symposium on Theoretical Aspects of Computer Science,
                  {STACS} 2026, Grenoble, France, March 9-13, 2026},
  series       = {LIPIcs},
  volume       = {364},
  pages        = {79:1--79:19},
  publisher    = {Schloss Dagstuhl - Leibniz-Zentrum f{\"{u}}r Informatik},
  year         = {2026},
  doi          = {10.4230/LIPICS.STACS.2026.79}
}

\end{document}